\documentclass[11pt]{article}

\usepackage[T1]{fontenc}
\usepackage{lmodern}
\usepackage{microtype}
\usepackage[margin=1in]{geometry}
\usepackage{amsmath,amssymb,amsthm,mathtools,bm}
\usepackage{booktabs,array}
\usepackage{graphicx}
\usepackage{enumitem}
\usepackage[section]{placeins}
\usepackage[numbers,sort&compress]{natbib}
\usepackage[hidelinks]{hyperref}
\usepackage{xcolor}

\hypersetup{
  pdftitle={Maximal-Rank Squeezing from Rank-One Quenches: An Exact Spectral-Reachability Theorem for Gaussian Bosonic Systems},
  pdfauthor={Kamran Ansari},
  pdfsubject={Exact finite-coupling theorem showing that one localized rank-one quadratic element can generate maximal-rank multimode squeezing},
  pdfkeywords={Bogoliubov transformation, multimode squeezing, quadratic Hamiltonian, low-rank perturbation, controllability Gramian, Krylov subspace, matrix square root}
}

\allowdisplaybreaks
\setlist{nosep,leftmargin=*}

\newtheorem{theorem}{Theorem}[section]
\newtheorem{proposition}[theorem]{Proposition}
\newtheorem{corollary}[theorem]{Corollary}
\newtheorem{lemma}[theorem]{Lemma}
\theoremstyle{definition}
\newtheorem{definition}[theorem]{Definition}
\theoremstyle{remark}
\newtheorem{remark}[theorem]{Remark}

\newcommand{\R}{\mathbb{R}}

\newcommand{\calC}{\mathcal{C}}
\newcommand{\calD}{\mathcal{D}}

\newcommand{\ran}{\operatorname{ran}}
\newcommand{\rank}{\operatorname{rank}}
\newcommand{\sgn}{\operatorname{sgn}}

\newcommand{\diag}{\operatorname{diag}}

\newcommand{\eps}{\varepsilon}
\newcommand{\norm}[1]{\left\lVert #1\right\rVert}
\newcommand{\inner}[2]{\left\langle #1,#2\right\rangle}

\title{\textbf{Maximal-Rank Squeezing from Rank-One Quenches:}\\
An Exact Spectral-Reachability Theorem for Gaussian Bosonic Systems}
\author{Kamran Ansari\\
\small Stanford University\\
\small \href{mailto:ansarik@stanford.edu}{\texttt{ansarik@stanford.edu}}}
\date{August 2026}

\begin{document}
\maketitle

\begin{abstract}
We find and prove that, for any finite $N$, a single localized rank-one quadratic element can activate all $N$ canonical squeezing channels at every nonzero stable coupling, provided the free spectrum is simple and the element overlaps every mode. Equivalently, a rank-one microscopic stiffness update can generate a full-rank anomalous Bogoliubov block. More generally, for free frequency matrix $\Omega>0$ and a stable sign-definite quench $\Omega^2\mapsto\Omega^2+\eps UU^T$,
\[
\rank\beta_\eps
=
\dim\operatorname{span}\{\ran U,\Omega\ran U,\Omega^2\ran U,\ldots\}
=
\sum_a\rank(P_aU),
\]
where $P_a$ projects onto each distinct free-frequency eigenspace. Thus the microscopic quench rank does not bound the squeezing rank; exact support is determined by spectral reachability. To the best of our knowledge, this exact all-coupling rank, kernel, and range law, including its one-element/full-rank corollary, has not been stated previously. The result identifies active and dark sectors, quantifies degeneracy bottlenecks, and gives an exact actuator-position law without computing the full Bogoliubov transformation. Its practical implication is that broad simultaneous squeezing support may require far fewer physical elements than active channels. A signed Lyapunov-Gramian path proves the result. Full algebraic rank certifies channel existence, not strength or independent tunability.
\end{abstract}

\noindent\textbf{Keywords:} Bogoliubov transformation; multimode squeezing; quadratic Hamiltonian; low-rank perturbation; controllability Gramian; Krylov subspace; matrix square root.

\section{Introduction}

Gaussian bosonic systems provide a common framework for continuous-variable quantum information, multimode quantum optics, cavity and circuit platforms, and dynamical Casimir protocols \citep{Braunstein2005,Weedbrook2012,Fabre2020,Johansson2010}. Because their Hamiltonians are quadratic in canonical variables, changes of oscillator structure can be represented by Bogoliubov transformations and analyzed through linear algebra. The Bloch--Messiah decomposition separates passive mode mixing from single-mode squeezing, and the nonzero singular values of the anomalous Bogoliubov block identify the active canonical squeezing channels \citep{Cariolaro2016}. Thus $\rank\beta_\eps$ is a basis-independent count of how many squeezing directions are present, before questions of strength, loss tolerance, or independent control are considered.

Many physical controls are much simpler than the multimode systems on which they act. A localized defect, a modulated boundary element, or a collective quadratic coordinate may introduce only one or a few spatial coupling profiles while overlapping many normal modes. After a finite mode truncation, a single term such as $\eps\phi(x_0)^2/2$ becomes a rank-one stiffness update $\eps uu^T$. This creates a distinction between the number of physical actuator profiles and the number of canonical response channels: the former is fixed by the update geometry, whereas the latter is determined only after the free spectrum and the resulting Bogoliubov transformation are taken into account.

The distinction arises because the final oscillator structure is governed by the principal matrix square root $W_\eps=(\Omega^2+\eps UU^T)^{1/2}$, not directly by the rank of $UU^T$. Distinct free frequencies process the same coupling direction differently, so a low-rank input can generate a dense, algebraically high-rank matrix-function correction. Low-rank matrix-function updates, square-root derivatives, and Krylov approximations are well studied in numerical linear algebra \citep{Beckermann2018,Fasi2023,Shmueli2024}. The Fr\'echet derivative of the principal square root is characterized by a Sylvester or Lyapunov equation and an integral Gramian representation \citep{DelMoral2018,Higham2008,Simoncini2016}. In mathematical physics, general Bogoliubov diagonalization and quadratic Hamiltonians with rank-one field perturbations are also established \citep{Nam2016,Derezinski2017,Gamet2026}. These results supply the main ingredients, but they do not by themselves give an exact finite-coupling count of which canonical squeezing channels are nonzero.

The question addressed here is therefore an exact support question: for a stable sign-definite finite-rank quench, what are the kernel, range, and rank of $\beta_\eps$, and can a direction that appears at weak coupling disappear through finite-coupling cancellation? If $\Omega=\sum_{a=1}^{m}\nu_aP_a$ is the decomposition into distinct free-frequency sectors, we prove
\begin{equation}
\boxed{
\rank\beta_\eps=\sum_{a=1}^{m}\rank(P_aU).
}
\label{eq:headline}
\end{equation}
Equivalently, the active-channel space is the frequency-reachable block-Krylov subspace generated by $(\Omega,U)$. A rank-one source contributes exactly one channel in every sector it touches and is full rank for a simple spectrum when every mode overlap is nonzero. The projections $P_aU$ identify dark sectors, quantify degeneracy bottlenecks, and determine actuator-placement effects before the full Gaussian transformation is computed.

We prove the result by factoring $\beta_\eps$ through a matrix-square-root correction. Its path derivative is a signed positive semidefinite controllability Gramian, so reachable directions cannot cancel at finite coupling. This upgrades a known first-order structure \citep{DelMoral2018,Higham2008,Simoncini2016} to an exact nonlinear support theorem and yields a standalone square-root kernel-and-range result. The result claimed as new is the exact all-coupling support law and its rank-one, degeneracy, and actuator-position consequences, not the qualitative possibility that one element can couple multiple modes.

\section{Gaussian quench and the meaning of squeezing rank}
\label{sec:setup}

Work in units with $\hbar=1$. Let $q,p\in\R^N$ be canonical quadratures satisfying $[q_j,p_k]=i\delta_{jk}$. The free Hamiltonian is
\begin{equation}
H_0=\frac12p^Tp+\frac12q^T\Omega^2q,
\qquad \Omega=\Omega^T>0.
\label{eq:H0}
\end{equation}
We apply a sign-definite finite-rank stiffness quench,
\begin{equation}
H_\eps=\frac12p^Tp+\frac12q^T\bigl(\Omega^2+\eps UU^T\bigr)q,
\qquad U\in\R^{N\times r},
\label{eq:Heps}
\end{equation}
where $\eps\in\R\setminus\{0\}$ and
\begin{equation}
K_\eps:=\Omega^2+\eps UU^T>0.
\label{eq:stability}
\end{equation}
No full-column-rank assumption is imposed on $U$. The stiffness update $UU^T$ has rank $d:=\dim\ran U\le r$. Positive and negative quenches are both permitted whenever the final stiffness remains positive definite. Define the final frequency matrix
\begin{equation}
W_\eps:=K_\eps^{1/2}>0.
\end{equation}

The annihilation operators associated with the initial and final oscillator structures may be written as
\begin{align}
a&=\frac{1}{\sqrt2}\left(\Omega^{1/2}q+i\Omega^{-1/2}p\right),\label{eq:a}\\
b&=\frac{1}{\sqrt2}\left(W_\eps^{1/2}q+iW_\eps^{-1/2}p\right).
\label{eq:b}
\end{align}
They are related by
\begin{equation}
b=\alpha_\eps a+\beta_\eps a^\dagger,
\end{equation}
where
\begin{align}
\alpha_\eps&=\frac12\left(W_\eps^{1/2}\Omega^{-1/2}+W_\eps^{-1/2}\Omega^{1/2}\right),\label{eq:alpha}\\
\beta_\eps&=\frac12\left(W_\eps^{1/2}\Omega^{-1/2}-W_\eps^{-1/2}\Omega^{1/2}\right).
\label{eq:beta}
\end{align}
A passive orthogonal rotation can diagonalize $W_\eps$ after this change of oscillator structure. Such passive left or right mode rotations do not change $\rank\beta_\eps$.

The Bloch--Messiah decomposition gives passive transformations $V_1,V_2$ and nonnegative squeezing parameters $s_1,\ldots,s_N$ such that, up to the usual complex-conjugation convention,
\begin{equation}
\beta_\eps=V_1\diag(\sinh s_1,\ldots,\sinh s_N)V_2^T.
\end{equation}
Consequently,
\begin{equation}
\rank\beta_\eps=\#\{j:s_j>0\}.
\label{eq:rankchannels}
\end{equation}
Thus the rank of the anomalous block is the number of nonzero canonical squeezing channels. It is invariant under changes of input and output mode basis.

The key elementary factorization is
\begin{equation}
\boxed{
\beta_\eps
=\frac12W_\eps^{-1/2}(W_\eps-\Omega)\Omega^{-1/2}.
}
\label{eq:betafactor}
\end{equation}
The outer factors are invertible, so
\begin{equation}
\rank\beta_\eps=\rank(W_\eps-\Omega).
\label{eq:rankfactor}
\end{equation}
The quantum question is therefore reduced to an exact support question for a matrix-square-root update: when can a rank-one $UU^T$ make $W_\eps-\Omega$, and hence $\beta_\eps$, full rank?

\section{Frequency reachability}
\label{sec:reachability}

\begin{definition}[Frequency-reachable subspace]
For $\Omega=\Omega^T$ and $U\in\R^{N\times r}$, define
\begin{equation}
\calC(\Omega,U)
:=
\operatorname{span}\{\ran U,\Omega\ran U,\Omega^2\ran U,\ldots\}.
\label{eq:Cdef}
\end{equation}
In finite dimension, the span stabilizes after at most $N$ powers.
\end{definition}

This is the controllable subspace of the stable linear pair $(-\Omega,U)$. It is also the block-Krylov subspace generated by the free frequency matrix and the coupling directions.

\begin{lemma}[Spectral-sector decomposition]
\label{lem:spectralC}
Let
\begin{equation}
\Omega=\sum_{a=1}^{m}\nu_aP_a
\label{eq:spectralOmega}
\end{equation}
be the spectral decomposition into distinct eigenvalues $\nu_a>0$ and orthogonal projectors $P_a$. Then
\begin{equation}
\boxed{
\calC(\Omega,U)=\bigoplus_{a=1}^{m}\ran(P_aU)
}
\label{eq:Cspectral}
\end{equation}
and
\begin{equation}
\boxed{
\dim\calC(\Omega,U)=\sum_{a=1}^{m}\rank(P_aU).
}
\label{eq:Cdim}
\end{equation}
\end{lemma}

\begin{proof}
Every vector $\Omega^kUz$ is a linear combination of the vectors $P_aUz$, so $\calC(\Omega,U)$ is contained in the right-hand side of \eqref{eq:Cspectral}. Conversely, because the eigenvalues $\nu_a$ are distinct, each $P_a$ is a polynomial in $\Omega$ by Lagrange interpolation. Hence $P_aU$ belongs to the block-Krylov span. Orthogonality of the spectral sectors gives the direct sum and the dimension formula.
\end{proof}

Our central mathematical finding is stronger than the dimension formula: at every stable finite coupling, $\calC(\Omega,U)$ is exactly the range of both the square-root correction and $\beta_\eps$, while $\calC(\Omega,U)^\perp$ is exactly their kernel.

\section{Exact finite-coupling support theorem}
\label{sec:main}

The following theorem is the main result. It gives the exact finite-coupling channel count and identifies the full kernel and range of the response.

\begin{theorem}[Exact spectral-reachability theorem]
\label{thm:main}
Let $\Omega=\Omega^T>0$, let $U\in\R^{N\times r}$, and let $\eps\in\R\setminus\{0\}$ satisfy $\Omega^2+\eps UU^T>0$. Set
\begin{equation}
W_\eps=(\Omega^2+\eps UU^T)^{1/2},
\qquad \calC=\calC(\Omega,U).
\end{equation}
Then
\begin{align}
\sgn(\eps)(W_\eps-\Omega)&\succeq0,\label{eq:signX}\\
\ker(W_\eps-\Omega)&=\calC^\perp,\label{eq:kernelX}\\
\ran(W_\eps-\Omega)&=\calC.\label{eq:rangeX}
\end{align}
The anomalous Bogoliubov block satisfies the stronger support identities
\begin{equation}
\boxed{
\ker\beta_\eps=\calC^\perp,
\qquad
\ran\beta_\eps=\calC,
}
\label{eq:betasupport}
\end{equation}
and therefore
\begin{equation}
\boxed{
\rank\beta_\eps
=
\dim\calC(\Omega,U)
=
\sum_{a=1}^{m}\rank(P_aU).
}
\label{eq:mainrank}
\end{equation}
In particular, $\sgn(\eps)(W_\eps-\Omega)$ is positive definite when restricted to $\calC$.
\end{theorem}

\begin{proof}
Introduce the path
\begin{equation}
K_t=\Omega^2+t\eps UU^T,
\qquad
W_t=K_t^{1/2},
\qquad 0\le t\le1.
\label{eq:path}
\end{equation}
The entire path is positive definite because
\begin{equation}
K_t=(1-t)\Omega^2+tK_\eps>0.
\end{equation}
The principal square root is analytic on the positive-definite cone, so $t\mapsto W_t$ is smooth. Differentiating $W_t^2=K_t$ gives the Lyapunov equation
\begin{equation}
W_t\dot W_t+\dot W_tW_t=\eps UU^T.
\label{eq:lyap}
\end{equation}
Since $W_t>0$, its unique solution is
\begin{equation}
\dot W_t
=
\eps\int_0^\infty e^{-sW_t}UU^Te^{-sW_t}\,ds.
\label{eq:gramianpath}
\end{equation}
Let $\eta=\sgn(\eps)$. The integrand in
\begin{equation}
\eta\dot W_t
=
|\eps|\int_0^\infty e^{-sW_t}UU^Te^{-sW_t}\,ds
\end{equation}
is positive semidefinite, and hence
\begin{equation}
\eta(W_\eps-\Omega)=\int_0^1\eta\dot W_t\,dt\succeq0.
\end{equation}
This proves \eqref{eq:signX}.

The subspace $\calC$ is invariant under $\Omega$ and contains $\ran U$. Because $\Omega$ is symmetric, $\calC^\perp$ is also invariant under $\Omega$, and $U^Tx=0$ for $x\in\calC^\perp$. Thus every $K_t$ reduces the orthogonal decomposition $\calC\oplus\calC^\perp$, and its restriction to $\calC^\perp$ is exactly $\Omega^2$. Functional calculus gives
\begin{equation}
W_t|_{\calC^\perp}=\Omega|_{\calC^\perp}.
\end{equation}
Therefore
\begin{equation}
\calC^\perp\subseteq\ker(W_\eps-\Omega).
\label{eq:firstinclusion}
\end{equation}

For the reverse inclusion, let $x\in\ker(W_\eps-\Omega)$. By the positivity established above,
\begin{equation}
0=x^T\eta(W_\eps-\Omega)x
=\int_0^1x^T\eta\dot W_tx\,dt.
\label{eq:zeroquad}
\end{equation}
The integrand is continuous and nonnegative, so it vanishes at every $t$, in particular at $t=0$. Since $W_0=\Omega$, \eqref{eq:gramianpath} yields
\begin{equation}
0
=|\eps|\int_0^\infty\norm{U^Te^{-s\Omega}x}_2^2\,ds.
\label{eq:zerogramian}
\end{equation}
The scalar integrand in \eqref{eq:zerogramian} is continuous and nonnegative, so it vanishes for every $s\ge0$. Hence $U^Te^{-s\Omega}x=0$ for all $s\ge0$. This vector-valued function is analytic in $s$, and differentiating at $s=0$ gives
\begin{equation}
U^T\Omega^kx=0
\qquad\text{for every }k\ge0.
\end{equation}
Equivalently, $x$ is orthogonal to every $\Omega^k\ran U$, so $x\in\calC^\perp$. Combined with \eqref{eq:firstinclusion}, this proves \eqref{eq:kernelX}.

The matrix $W_\eps-\Omega$ is symmetric, reduces $\calC\oplus\calC^\perp$, vanishes on $\calC^\perp$, and has no kernel on $\calC$. Its range is therefore $\calC$, proving \eqref{eq:rangeX}.

Finally, $\Omega$ and $W_\eps$ both preserve $\calC$ and $\calC^\perp$; hence so do all of their real powers, by functional calculus. From the factorization \eqref{eq:betafactor},
\begin{equation}
\beta_\eps x=0
\iff
\Omega^{-1/2}x\in\ker(W_\eps-\Omega)
\iff
x\in\calC^\perp,
\end{equation}
which proves the kernel identity for $\beta_\eps$. Its range is
\begin{equation}
\ran\beta_\eps
=W_\eps^{-1/2}\ran(W_\eps-\Omega)
=W_\eps^{-1/2}\calC
=\calC.
\end{equation}
Taking dimensions and applying Lemma~\ref{lem:spectralC} proves \eqref{eq:mainrank}.
\end{proof}

\subsection{Immediate consequences of the main finding}

\begin{corollary}[Exact rank-one channel count]
\label{cor:rankone}
Let $U=u\in\R^N$. Then
\begin{equation}
\boxed{
\rank\beta_\eps=\#\{a:P_au\neq0\}.
}
\label{eq:rankonecount}
\end{equation}
A rank-one source contributes exactly one squeezing channel in every distinct frequency sector that it touches.
\end{corollary}

\begin{corollary}[One localized rank-one element activates all channels]
\label{cor:generic}
If $\Omega$ has simple spectrum and $u$ has nonzero overlap with every eigenvector of $\Omega$, then
\begin{equation}
\rank\beta_\eps=N
\end{equation}
for every nonzero stable coupling. For fixed simple $\Omega$, this condition holds for Lebesgue-almost every $u\in\R^N$. Thus one generic localized rank-one quadratic element activates every canonical squeezing channel in the retained $N$-mode system.
\end{corollary}

\begin{corollary}[Degeneracy bottleneck]
\label{cor:degenerate}
If the frequency $\nu_a$ has multiplicity $g_a$, its exact contribution to the squeezing rank is
\begin{equation}
\rank(P_aU)\le\min(g_a,d)\le\min(g_a,r),
\qquad d=\rank U.
\end{equation}
For a generic $U\in\R^{N\times r}$ (which has $d=\min(N,r)$),
\begin{equation}
\rank\beta_\eps=\sum_a\min(g_a,r).
\label{eq:genericdegenerate}
\end{equation}
In particular, a rank-one quench resolves at most one direction inside each degenerate frequency sector.
\end{corollary}

\begin{remark}[Depth of spectral mixing, not bare rank]
A low-rank quench has a low-rank response only when its coupling directions remain confined to a small invariant subspace of $\Omega$. If $U$ lies in one eigenspace, then no rank amplification occurs. If the coupling has support across many distinct spectral sectors, the response rank grows to the dimension of the resulting block-Krylov space.
\end{remark}

\subsection{A standalone matrix-square-root theorem}

We also obtain an independent matrix-analysis finding: the exact support of a sign-definite square-root update.

\begin{corollary}[Exact support of a sign-definite square-root update]
\label{cor:sqrt}
Let $A=A^T>0$, let $U\in\R^{N\times r}$, and let $\eps\in\R\setminus\{0\}$ satisfy $A+\eps UU^T>0$. Define
\begin{equation}
\calD(A,U):=\operatorname{span}\{\ran U,A\ran U,A^2\ran U,\ldots\}.
\label{eq:Ddef}
\end{equation}
Then
\begin{align}
\sgn(\eps)\bigl((A+\eps UU^T)^{1/2}-A^{1/2}\bigr)&\succeq0,\label{eq:sqrtsign}\\
\ker\bigl((A+\eps UU^T)^{1/2}-A^{1/2}\bigr)&=\calD(A,U)^\perp,\label{eq:sqrtkernel}\\
\ran\bigl((A+\eps UU^T)^{1/2}-A^{1/2}\bigr)&=\calD(A,U).\label{eq:sqrtrange}
\end{align}
Consequently,
\begin{equation}
\boxed{
\rank\left((A+\eps UU^T)^{1/2}-A^{1/2}\right)
=
\dim\calD(A,U).
}
\label{eq:sqrtrank}
\end{equation}
\end{corollary}

\begin{proof}
Apply Theorem~\ref{thm:main} with $\Omega=A^{1/2}$. The matrices $A$ and $A^{1/2}$ have the same spectral projectors, so their cyclic subspaces generated by $U$ coincide with $\calD(A,U)$.
\end{proof}

This identity sharpens the general observation that a low-rank input update can generate a full-rank matrix-function correction. It also explains why full algebraic rank and accurate low-rank approximation can coexist: the theorem determines which singular values are nonzero, not how large they are.

\section{Localized defects and an exact actuator-position rule}
\label{sec:cavity}

Applying the theorem to a localized field defect, we find an exact actuator-position law. Consider a scalar field after a finite spectral truncation,
\begin{equation}
\phi(x)=\sum_{n=1}^{N}q_n\varphi_n(x),
\qquad
H_0=\frac12\sum_{n=1}^{N}\left(p_n^2+\omega_n^2q_n^2\right).
\label{eq:fieldexp}
\end{equation}
A localized quadratic defect at $x_0$ gives
\begin{equation}
\frac{\eps}{2}\phi(x_0)^2
=
\frac{\eps}{2}q^Tuu^Tq,
\qquad
u_n=\varphi_n(x_0).
\label{eq:pointdefect}
\end{equation}
The bare stiffness update is rank one. If the cavity spectrum is nondegenerate, Corollary~\ref{cor:rankone} gives
\begin{equation}
\rank\beta_\eps
=
\#\{n\le N:\varphi_n(x_0)\neq0\}.
\label{eq:fieldcount}
\end{equation}
This is the localized-element consequence of the theorem: one point defect activates exactly the retained frequency sectors whose mode functions are nonzero at the defect position.

\begin{proposition}[One-dimensional Dirichlet cavity]
\label{prop:dirichlet}
For a massless scalar field on $[0,L]$ with Dirichlet modes
\begin{equation}
\varphi_n(x)=\sqrt{\frac{2}{L}}\sin\left(\frac{n\pi x}{L}\right),
\qquad
\omega_n=\frac{n\pi c}{L},
\end{equation}
consider the first $N$ modes and a stable point-defect quench at $x_0\in(0,L)$. Then:
\begin{enumerate}[label=(\roman*)]
\item If $x_0/L$ is irrational, $\rank\beta_\eps=N$.
\item If $x_0/L=p/q$ in lowest terms, then
\begin{equation}
\boxed{
\rank\beta_\eps=N-\left\lfloor\frac{N}{q}\right\rfloor.
}
\label{eq:rationalplacement}
\end{equation}
\end{enumerate}
\end{proposition}

\begin{proof}
The spectrum is simple. For irrational $x_0/L$, $n x_0/L$ is never an integer, so no mode amplitude vanishes. If $x_0/L=p/q$ in lowest terms, then $\sin(n\pi p/q)=0$ exactly when $q$ divides $n$. There are $\lfloor N/q\rfloor$ such modes among the first $N$, and Corollary~\ref{cor:rankone} gives the result.
\end{proof}

For example, a defect at the midpoint $x_0=L/2$ misses every even mode and activates $\lceil N/2\rceil$ channels. More generally, for any fixed cutoff $N$, all but finitely many positions in $(0,L)$ make one point defect full rank and therefore produce all $N$ nonzero canonical squeezing channels; the only exceptions are the finitely many nodes of the first $N$ mode functions. In higher-dimensional cavities, symmetry-induced degeneracies introduce the additional bottleneck in Corollary~\ref{cor:degenerate}.

A literal point interaction can require ultraviolet regularization in the continuum. A smeared field operator $\phi(f)=\sum_n\inner{f}{\varphi_n}q_n$ gives the same finite-mode theorem with $u_n=\inner{f}{\varphi_n}$. The exact continuum limit requires the self-adjointness, stability, and implementability conditions studied for bosonic quadratic Hamiltonians and rank-one field perturbations \citep{Nam2016,Derezinski2017,Gamet2026,Shale1962}.

\section{Design implications, conditioning, and numerical checks}
\label{sec:applications}

The main physical finding is a rank-amplification law: the hardware rank $d=\rank U$ is not, by itself, the number of squeezing channels. The exact response rank is
\begin{equation}
\boxed{
r_{\mathrm{squeeze}}=\sum_a\rank(P_aU).
}
\label{eq:designrule}
\end{equation}
This gives a direct hardware-screening rule before the full Gaussian transformation is computed.

\paragraph{One physical element can activate every canonical channel.}
Under the simple-spectrum overlap condition, we find that a localized defect, tunable boundary element, or collective quadratic coordinate can be rank one in the Hamiltonian while producing all $N$ nonzero canonical squeezing channels. Thus broad simultaneous support need not require an independent element on every channel, a potentially useful reduction in cavity and circuit hardware \citep{Johansson2010}.

\paragraph{Geometry and spectral overlap become design variables.}
The projections $P_aU$ identify exactly which sectors are dark. In a simple-spectrum cavity, avoiding the finitely many nodal positions of the retained modes is enough to obtain full algebraic support. With degeneracies, an actuator of physical rank $d$ can resolve at most $d$ directions per degenerate sector, so lifting degeneracies, changing actuator geometry, or adding independent coupling profiles can enlarge the support.

\paragraph{Support and strength must be optimized separately.}
At weak coupling, the signed square-root derivative is governed by the controllability Gramian
\begin{equation}
G_0=\int_0^\infty e^{-s\Omega}UU^Te^{-s\Omega}\,ds,
\label{eq:G0}
\end{equation}
and the first-order anomalous block is $\frac12\Omega^{-1/2}G_0\Omega^{-1/2}$ up to the coupling scale. In a simple-spectrum rank-one eigenbasis,
\begin{equation}
\left.\frac{d\beta_\eps}{d\eps}\right|_{\eps=0,ij}
=
\frac{u_iu_j}{2\sqrt{\omega_i\omega_j}(\omega_i+\omega_j)}.
\label{eq:weakbeta}
\end{equation}
The Cauchy kernel $1/(\omega_i+\omega_j)$ is nonsingular for distinct positive frequencies, explaining how a rank-one numerator spreads across every touched frequency; Theorem~\ref{thm:main} shows that this support persists at every stable nonzero coupling. A useful design should therefore first maximize the exact support dimension and then optimize a conditioning metric, such as the smallest eigenvalue of $G_0$ on $\calC$ or the smallest targeted singular value of $\beta_\eps$. Multimode squeezing experiments already use spectral and spatial shaping to tailor mode strengths \citep{Arzani2018,Fabre2020}; the theorem adds an exact support criterion.

\subsection{Numerical consistency and effective rank}
\label{sec:numerics}

We evaluated the exact matrices by symmetric eigendecomposition and compared thresholded numerical ranks with the spectral prediction in \eqref{eq:mainrank}. These calculations check implementation and conditioning; they are not substitutes for the proof. For the observed column of Table~\ref{tab:numerics}, we use the strict relative threshold $\tau_{\mathrm{num}}=10^{-12}$ defined below; all matrices and coupling constants are specified in Appendix~\ref{app:numerics}.

\begin{table}[htbp]
\centering
\caption{Finite-dimensional examples. The final column is the ratio of the smallest predicted-active singular value to the largest singular value of $\beta_\eps$. The observed rank uses the strict relative threshold $\tau_{\mathrm{num}}=10^{-12}$.}
\label{tab:numerics}
\begin{tabular}{@{}lccccc@{}}
\toprule
case & modes $N$ & bare rank & predicted & observed & $\sigma_{\min}/\sigma_{\max}$\\
\midrule
rank-one, simple spectrum & 8 & 1 & 8 & 8 & $1.72\times10^{-9}$\\
rank-three, mixed profiles & 10 & 3 & 10 & 10 & $3.97\times10^{-5}$\\
degenerate, partial support & 9 & 2 & 5 & 5 & $4.08\times10^{-4}$\\
\bottomrule
\end{tabular}
\end{table}

Across the three named examples, all factorization, Lyapunov, and 160-node Gauss--Legendre path-integration residuals were below $2\times10^{-14}$. A fixed-seed suite of 250 additional positive and negative stable quenches produced no rank mismatches at threshold $10^{-12}$, and its monitored factorization and Lyapunov residuals were also below $2\times10^{-14}$. Exact floating-point values and execution-environment metadata are stored in the public result record because the final digits can vary slightly across BLAS/LAPACK implementations.

\begin{figure}[htbp]
\centering
\includegraphics[width=0.78\linewidth]{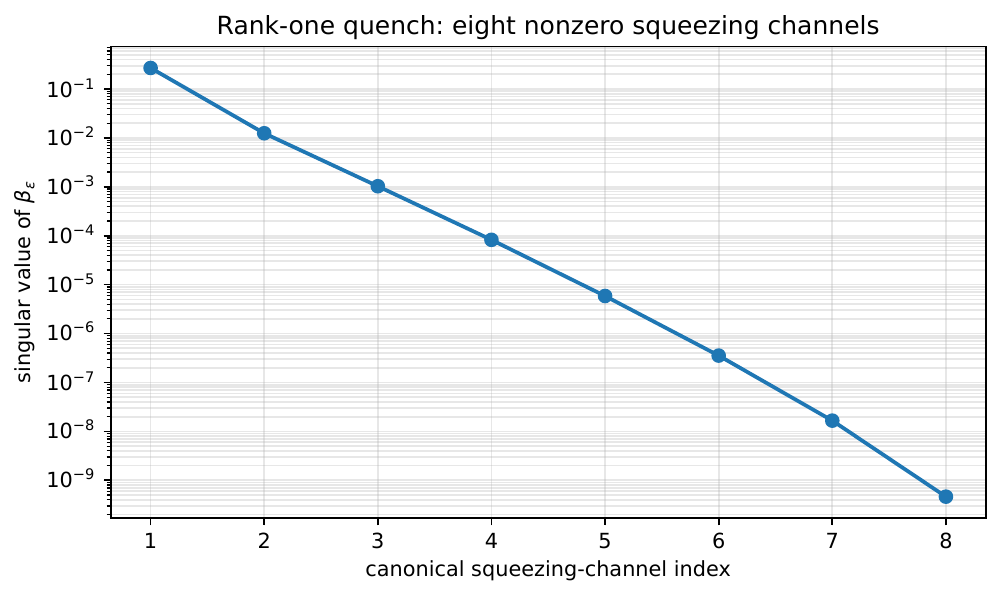}
\caption{Singular values of $\beta_\eps$ for the eight-mode rank-one example in Appendix~\ref{app:numerics}. Theorem~\ref{thm:main} certifies eight nonzero squeezing channels, but their strengths span nearly nine orders of magnitude.}
\label{fig:spectrum}
\end{figure}

Figure~\ref{fig:spectrum} illustrates the distinction between exact and practical rank. Algebraic rank answers whether a channel is exactly absent; it does not determine whether that channel is experimentally resolvable, robust to loss, or useful on a finite time scale. If $\beta_\eps\neq0$, define, for $0\le\tau<1$,
\begin{equation}
r_{\mathrm{eff}}(\tau)
:=
\#\{j:\sigma_j(\beta_\eps)>\tau\sigma_1(\beta_\eps)\}.
\label{eq:effectiverank}
\end{equation}
If $\beta_\eps=0$, set $r_{\mathrm{eff}}(\tau)=0$. The strict inequality gives
\begin{equation}
r_{\mathrm{eff}}(0)=\rank\beta_\eps,
\end{equation}
which Theorem~\ref{thm:main} determines exactly. Estimating $r_{\mathrm{eff}}(\tau)$ for $\tau>0$ is a separate conditioning problem, consistent with the rapid spectral decay studied for low-rank square-root corrections \citep{Shmueli2024}.

The finding concerns simultaneous support, not arbitrary high-dimensional control. A single quench parameter traces only a one-parameter family of Gaussian transformations, with channel strengths and mode shapes fixed by $\Omega$, $U$, and $\eps$. Arbitrary synthesis generally requires additional controls, time sequencing, and passive mixing. The theorem is therefore best used as a first-stage hardware screen: determine whether one or a few elements already reach the desired channel set, then optimize strength and controllability within that support.

\section{Prior work, scope, and limitations}
\label{sec:scope}

The broad physical fact that one localized or collective element can couple many bosonic modes is not claimed as new. General Bogoliubov diagonalization and the operator-theoretic treatment of quadratic bosonic Hamiltonians are established in finite and infinite dimensions \citep{Nam2016,Derezinski2017}. Gamet \citep{Gamet2026} studies a free bosonic Hamiltonian plus the square of a smeared field operator, with emphasis on diagonalization, self-adjointness, ultraviolet behavior, and renormalization. In numerical linear algebra, low-rank matrix-function updates, square-root derivatives, and Krylov approximations are also established \citep{DelMoral2018,Fasi2023,Beckermann2018,Shmueli2024,Higham2008,Simoncini2016}.

Against this background, the finding claimed here is narrower and exact: the kernel, range, and rank of the stable sign-definite square-root update, and therefore of the anomalous Bogoliubov response, are determined solely by spectral reachability at every finite stable coupling. To the best of our knowledge, the identity \eqref{eq:mainrank}, its kernel-and-range strengthening \eqref{eq:betasupport}, the one-element/full-rank corollary, and the actuator-position count \eqref{eq:rationalplacement} have not previously been stated in this form. The novelty claim is limited to these exact support statements and their Gaussian-bosonic interpretation; the underlying Hamiltonian, Bogoliubov, matrix-square-root, Gramian, Cauchy, and approximation ingredients are not claimed as new.

The assumptions are essential. The quench is stable, sign definite, and instantaneous. Indefinite updates can lose the common semidefinite sign that prevents cancellation, and general time-dependent pulses require a propagator-level analysis; a slow adiabatic ramp need not create particles. The theorem certifies algebraic support rather than particle number, entanglement, metrological gain, loss tolerance, or signal-to-noise ratio. It is also finite dimensional and Gaussian: an infinite-dimensional field theory additionally requires well-defined operator or quadratic-form dynamics, stability, and the appropriate Hilbert--Schmidt implementability condition \citep{Shale1962,Nam2016,Derezinski2017,Gamet2026}, while nonquadratic interactions can invalidate the Bogoliubov description. Extending the support law to indefinite, time-dependent, dissipative, and infinite-dimensional settings remains open.

\section{Conclusion}

We find that one localized rank-one quadratic element can activate all $N$ canonical squeezing channels in a finite $N$-mode system, and we prove that this holds at every nonzero stable coupling under the simple-spectrum and all-mode-overlap conditions. The microscopic update has rank one; the exact anomalous Bogoliubov response has rank $N$. More generally,
\[
\rank\beta_\eps=\sum_a\rank(P_aU),
\]
so the active-channel space is precisely the spectral-reachable subspace generated by $(\Omega,U)$. In a one-dimensional Dirichlet cavity, this gives the exact placement law $N-\lfloor N/q\rfloor$ for a defect at $x_0/L=p/q$ in lowest terms. These exact finite-coupling support laws are the central findings.

The practical implication is a two-stage design rule: first use the unperturbed spectrum and actuator profile to identify dark sectors, resolve degeneracy bottlenecks, and determine whether one or a few elements already provide the desired support; then optimize singular values, pulse sequences, and independent control. The theorem does not imply that every active channel is strong, loss tolerant, or separately tunable. It establishes the potentially hardware-saving result that physical perturbation rank can be far smaller than simultaneous multimode response rank.

\section*{Data and code availability}
No external datasets are used. The complete manuscript source, verification code, figure-generation script, and machine-readable numerical results are publicly available at \url{https://github.com/kansari123/Exact-Spectral-Reachability-Theorem-for-Gaussian-Bosonic-Systems}. The repository reproduces Table~\ref{tab:numerics}, Figure~\ref{fig:spectrum}, and the checks in Section~\ref{sec:numerics}.

\appendix

\section{Numerical methods and reproducibility}
\label{app:numerics}

The supplied script computes principal powers of a real symmetric positive-definite matrix by eigendecomposition. Given $\Omega,U,\eps$, it forms
\begin{equation}
K_\eps=\Omega^2+\eps UU^T,
\qquad
W_\eps=K_\eps^{1/2},
\end{equation}
and evaluates \eqref{eq:beta} and \eqref{eq:betafactor} independently. The predicted rank is computed from spectral projectors of $\Omega$, not from an ill-conditioned raw Krylov matrix. Within a degenerate block, the contribution is the thresholded numerical rank of $P_aU$, using the strict criterion $\sigma_j>10^{-12}\sigma_1$.

The three examples in Table~\ref{tab:numerics} are fully deterministic:
\begin{align*}
\text{(i)}\quad &N=8,\quad \Omega=\diag(10^{(j-1)/7})_{j=1}^{8},\quad U=(1,\ldots,1)^T,\quad \eps=1;\\
\text{(ii)}\quad &N=10,\quad \Omega=\diag(10^{(j-1)/9})_{j=1}^{10},\quad \eps=1,\\[-0.2em]
&U_{jk}=\sin(jk)+0.3\cos((j+1)(k+2)),\quad 1\le j\le10,\ 1\le k\le3;\\
\text{(iii)}\quad &N=9,\quad \Omega=\diag(1,1,1,2,2,3,3,3,3),\quad \eps=\tfrac12,\\[-0.2em]
&U^T=\begin{pmatrix}
1&0&1&1&2&1&0&1&2\\
0&1&1&2&4&0&1&-1&1
\end{pmatrix}.
\end{align*}
The projected ranks in case (iii) are $(2,1,2)$, giving exact rank five.

For each case, the script integrates \eqref{eq:gramianpath} over $t\in[0,1]$ with 160-node Gauss--Legendre quadrature. It also runs 250 fixed-seed trials (seed $20260812$), alternating positive and negative stable quenches, and checks the spectral-sector rank at relative threshold $10^{-12}$, the Lyapunov residual \eqref{eq:lyap}, and the factorization \eqref{eq:betafactor}. Residuals are Frobenius norms divided by the larger norm of the compared sides. Exact arrays, singular values, tolerances, residuals, and the suite summary are stored in \texttt{results/verification\_results.json}. Severe Cauchy conditioning can push the weakest nonzero singular values below numerical precision; this affects thresholded rank detection, not the exact theorem.


{\footnotesize
\begin{thebibliography}{99}
\setlength{\itemsep}{0.05em}

\bibitem{Braunstein2005}
S.~L. Braunstein and P. van Loock,
``Quantum information with continuous variables,''
\emph{Reviews of Modern Physics} \textbf{77}, 513--577 (2005).
\href{https://doi.org/10.1103/RevModPhys.77.513}{doi:10.1103/RevModPhys.77.513}.

\bibitem{Weedbrook2012}
C. Weedbrook, S. Pirandola, R. Garc\'{i}a-Patr\'{o}n, N.~J. Cerf,
T.~C. Ralph, J.~H. Shapiro, and S. Lloyd,
``Gaussian quantum information,''
\emph{Reviews of Modern Physics} \textbf{84}, 621--669 (2012).
\href{https://doi.org/10.1103/RevModPhys.84.621}{doi:10.1103/RevModPhys.84.621}.

\bibitem{Fabre2020}
C. Fabre and N. Treps,
``Modes and states in quantum optics,''
\emph{Reviews of Modern Physics} \textbf{92}, 035005 (2020).
\href{https://doi.org/10.1103/RevModPhys.92.035005}{doi:10.1103/RevModPhys.92.035005}.

\bibitem{Cariolaro2016}
G. Cariolaro and G. Pierobon,
``Bloch--Messiah reduction of Gaussian unitaries by Takagi factorization,''
\emph{Physical Review A} \textbf{94}, 062109 (2016).
\href{https://doi.org/10.1103/PhysRevA.94.062109}{doi:10.1103/PhysRevA.94.062109}.

\bibitem{Nam2016}
P.~T. Nam, M. Napi\'{o}rkowski, and J.~P. Solovej,
``Diagonalization of bosonic quadratic Hamiltonians by Bogoliubov transformations,''
\emph{Journal of Functional Analysis} \textbf{270}, 4340--4368 (2016).
\href{https://doi.org/10.1016/j.jfa.2015.12.007}{doi:10.1016/j.jfa.2015.12.007}.

\bibitem{Derezinski2017}
J. Derezi\'{n}ski,
``Bosonic quadratic Hamiltonians,''
\emph{Journal of Mathematical Physics} \textbf{58}, 121101 (2017).
\href{https://doi.org/10.1063/1.5017931}{doi:10.1063/1.5017931}.

\bibitem{Gamet2026}
T. Gamet,
``Renormalization of Bosonic Quadratic Hamiltonians Involving Rank One Perturbations,''
\emph{Annales Henri Poincar\'e} (2026), published online 2 February 2026.
\href{https://doi.org/10.1007/s00023-026-01659-2}{doi:10.1007/s00023-026-01659-2}.

\bibitem{DelMoral2018}
P. Del Moral and A. Niclas,
``A Taylor expansion of the square root matrix function,''
\emph{Journal of Mathematical Analysis and Applications} \textbf{465}, 259--266 (2018).
\href{https://doi.org/10.1016/j.jmaa.2018.05.005}{doi:10.1016/j.jmaa.2018.05.005}.


\bibitem{Fasi2023}
M. Fasi, N.~J. Higham, and X. Liu,
``Computing the square root of a low-rank perturbation of the scaled identity matrix,''
\emph{SIAM Journal on Matrix Analysis and Applications} \textbf{44}, 156--174 (2023).
\href{https://doi.org/10.1137/22M1471559}{doi:10.1137/22M1471559}.

\bibitem{Beckermann2018}
B. Beckermann, D. Kressner, and M. Schweitzer,
``Low-rank updates of matrix functions,''
\emph{SIAM Journal on Matrix Analysis and Applications} \textbf{39}, 539--565 (2018).
\href{https://doi.org/10.1137/17M1140108}{doi:10.1137/17M1140108}.

\bibitem{Shmueli2024}
S. Shmueli, P. Drineas, and H. Avron,
``Low-rank updates of matrix square roots,''
\emph{Numerical Linear Algebra with Applications} \textbf{31}, e2528 (2024).
\href{https://doi.org/10.1002/nla.2528}{doi:10.1002/nla.2528}.

\bibitem{Higham2008}
N.~J. Higham,
\emph{Functions of Matrices: Theory and Computation}
(SIAM, Philadelphia, 2008).
\href{https://doi.org/10.1137/1.9780898717778}{doi:10.1137/1.9780898717778}.

\bibitem{Simoncini2016}
V. Simoncini,
``Computational methods for linear matrix equations,''
\emph{SIAM Review} \textbf{58}, 377--441 (2016).
\href{https://doi.org/10.1137/130912839}{doi:10.1137/130912839}.

\bibitem{Shale1962}
D. Shale,
``Linear symmetries of free boson fields,''
\emph{Transactions of the American Mathematical Society} \textbf{103}, 149--167 (1962).
\href{https://doi.org/10.2307/1993745}{doi:10.2307/1993745}.

\bibitem{Johansson2010}
J.~R. Johansson, G. Johansson, C.~M. Wilson, and F. Nori,
``The dynamical Casimir effect in superconducting microwave circuits,''
\emph{Physical Review A} \textbf{82}, 052509 (2010).
\href{https://doi.org/10.1103/PhysRevA.82.052509}{doi:10.1103/PhysRevA.82.052509}.

\bibitem{Arzani2018}
F. Arzani, C. Fabre, and N. Treps,
``Versatile engineering of multimode squeezed states by optimizing the pump spectral profile in spontaneous parametric down-conversion,''
\emph{Physical Review A} \textbf{97}, 033808 (2018).
\href{https://doi.org/10.1103/PhysRevA.97.033808}{doi:10.1103/PhysRevA.97.033808}.

\end{thebibliography}
}
\end{document}